\documentclass[journal, draftcls,onecolumn,12pt,twoside]{IEEEtranTCOM}
\usepackage[footnote,marginclue,nomargin,draft]{fixme}

\usepackage{graphics}
\usepackage{float}				
\usepackage{tikz}				
\usepackage{amssymb}
\usepackage{bm}					
\usepackage{sistyle}			
\SIproductsign{\cdot}
\usepackage{bbm}

\normalsize

\ifCLASSINFOpdf
\else
\fi

\usepackage[cmex10]{amsmath}
\usepackage{algorithmic}
\usepackage[tight,footnotesize]{subfigure}

\usepackage{stfloats}
\usepackage[utf8]{inputenc}	

\newcommand{\vectornorm}[1]{\left|\left|#1\right|\right|}

\usepackage{algorithm}
\usepackage{amsthm}

\newtheorem{lemma}{Lemma}

\newtheorem{proposition}{Proposition}

\theoremstyle{remark}
\newtheorem*{rem}{Remark}

\begin{document}

%
\title{Blahut-Arimoto Algorithm and Code Design for Action-Dependent Source
  Coding Problems}
%
%
%

\author{Kasper~Fløe~Trillingsgaard,~Osvaldo~Simeone,~Petar~Popovski~and~Torben~Larsen}
\maketitle

\begin{abstract}
  The source coding problem with action-dependent side information at the decoder
  has recently been introduced to model data acquisition in
  resource-constrained systems. In this paper, an efficient algorithm for
  numerical computation of the rate-distortion-cost function for this problem is
  proposed, and a convergence proof is provided. Moreover, a two-stage code
  design based on multiplexing is put forth, whereby the first stage encodes the actions and
  the second stage is composed of an array of classical Wyner-Ziv codes, one for
  each action. Specific coding/decoding strategies are designed based on LDGM
  codes and message passing. Through numerical examples, the proposed code
  design is shown to achieve performance close to the lower bound dictated by
  the rate-distortion-cost function.

\end{abstract}


%
\IEEEpeerreviewmaketitle

\section{Introduction}
%
%
%
\IEEEPARstart{T}{he} source coding problem in which the decoder can take actions
that affect the availability or quality of the side information at the decoder was
introduced in \cite{permuter}. The problem generalizes the well-known Wyner-Ziv
set-up and can be used to model data acquisition in resource-constrainted systems, such as
sensor networks. In the model studied in \cite{permuter}, each
action is associated a cost and the system design is subject to an average cost
constraint. The information-theoretic analysis of the problem was fully
addressed in \cite{permuter}. In this paper, instead, we tackle the practical
open issues, namely the computation of the rate-distortion-cost function and
code design.

Specifically, the rate-distortion-cost function for the source coding problem with
action-dependent side information was derived in \cite{permuter}. However, no
specific algorithm was proposed for its computation. A first contribution of this paper is
to propose such an algorithm by generalizing the classical Blahut-Arimoto (BA)
approach, which was introduced for the Wyner-Ziv problem in
\cite{dupuis}. Convergence of the algorithm is also proved. 

Moreover, while the theory in \cite{permuter} demonstrates the existence of coding and decoding strategies able
to achieve the rate-distortion-cost bound, practical code constructions have not
been investigated yet. It is recalled that, for classical lossy source coding problems, codes that have been
able to achieve rate-distortion bound include Low Density Generator Matrix
(LDGM) codes \cite{iterative_quantization}, polar codes \cite{polar_sc} and
trellis-based quantization codes \cite{trellis_based_quantization}.
For the Wyner-Ziv problem, efficient codes include compound LDPC/LDGM codes
\cite{ldgm_wyner} and polar codes \cite{polar_sc}.  A second contribution of
this paper is hence the study of code design for source coding problems with
action-dependent side information. As shown in \cite{permuter}, optimal codes
for this problem have a successive refinement structure, in which the first
layer produces the action sequence and the refinement layer uses binning to
leverage the side information at the decoder. Here, we first observe that a layered code structure in
which the refinement layer uses a multiplexing of separate classical Wyner-Ziv
codes, one for each action, is optimal. This allows us to simplify the code
structure with respect to the successive refinement strategy in
\cite{permuter}. LDGM-based codes with message passing encoding are designed and
demonstrated via numerical results to perform close to the rate-distortion-cost
function.

The paper is organized as follows. In Section~\ref{sec:background}, the
action-dependent source coding problem is described and results from
\cite{permuter} are summarized. In Section~\ref{sec:arimoto}, we describe the proposed
algorithm for computation of the rate-distortion-cost function, and in
Section~\ref{sec:code_design}, a practical code design is proposed. Finally, in
Section~\ref{sec:examples}, we present numerical results for a specific example.

\subsection{Notation}
Throughout this work, we let upper case, lower case and calligraphic letters
denote random variables, values and alphabets of the random variables, respectively. For jointly
distributed random variables, $P_X(x)$, $P_{X|Y}(x|y)$ and $P_{X,Y}(x,y)$ denote the
probability mass function (pmf) of
$X$, the conditional pmf of $X$ given $Y$ and the joint pmf of $X$ and $Y$. To
simplify notation, the subscripts of the pmfs may be omitted, e.g., 
$P(x|y)$ may be used instead of $P_{X|Y}(x|y)$. The notation $X^n$ represents the tuple $(X_1, X_2,
\ldots, X_n)$, and $[a,b]$ where $a,b\in\mathbb{Z}$ with $a<b$ denotes the set of integers
$\{a,a+1,\ldots, b-1,b\}$.  Moreover, $\mathbb{Z}_+=\{0,1,\ldots\}$,
$\mathbb{N}=\mathbb{Z}_+ \setminus \{0\}$ and $\mathbbm{1}_{\left\{\verb|cond|\right\}}$ denotes the indicator function, and is
one when \verb|cond| is true, and zero otherwise.
The notation $\lfloor \cdot \rfloor$ and $\lceil \cdot \rceil$ denotes the
floor and ceiling operators, respectively.

\section{Background}
\label{sec:background}
In this section, we recall the definition of source coding problems with
action-dependent side information and review the rate-distortion-cost function
obtained in \cite{permuter}.
\subsection{System Model}
\label{sec:system_model}
\begin{figure}[!t]
  \center
  \includegraphics[width=2.5in]{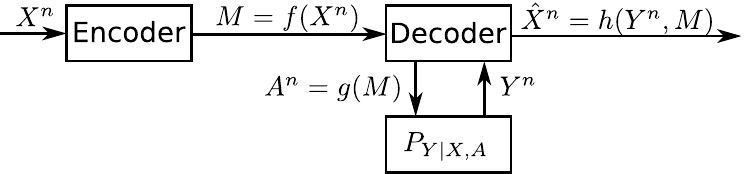}
  \caption{Source coding with action-dependent side information.}
  \label{fig:action_model}
\end{figure}
The source coding problem with action-dependent side
information introduced in
\cite{permuter} is illustrated in Fig.~\ref{fig:action_model}. In this problem, the
source $X^n\in\mathcal{X}^n$ is memoryless and each sample is distributed
according to the pmf $P_X$. At the encoder, the encoding function 
\begin{align}
f: \mathcal{X}^n\rightarrow \left[1,\lfloor 2^{nR}\rfloor \right],
\end{align}
  maps the source $X^n$ into a message $M\in\left[1,\lfloor
    2^{nR}\rfloor\right]$, where $R$ denotes the rate in bits per sample.
At the decoder, an action sequence
$A^n\in\mathcal{A}^n$ is chosen according to an action strategy 
\begin{align}
g:\left[1,\lfloor 2^{nR}\rfloor\right]\rightarrow \mathcal{A}^n,
\end{align}
which maps the message $M$ into an action sequence $A^n$. Based on $A^n$, the
side information $Y^n\in\mathcal{Y}^n$ is conditionally independent and
identically distributed (iid) according to the conditional pmf
$P_{Y|X,A}$ 
so that we have
\begin{align}
  P_{Y^n|X^n,A^n}(y^n|x^n,a^n)=\prod_{i=1}^n   P_{Y|X,A}(y_i|x_i,a_i).
\end{align}
The decoder makes a reconstruction $\hat X^n\in\mathcal{\hat X}^n$ of $X^n$ according to
the decoding function 
\begin{align}
h: \left[1,\lfloor 2^{nR}\rfloor\right]\times \mathcal{Y}^n\rightarrow
\mathcal{\hat X}^n,
\end{align}
which maps message $M$ and side information $Y^n$ into the estimate $\hat X^n$.

The action cost function $\Delta(a):\mathcal{A}\rightarrow \mathbb{R}_{+}$ is defined
such that $\Delta(a)=0$ for some $a\in \mathcal{A}$ and
$\Delta_{\text{max}}=\max_{a\in\mathcal{A}} \Delta(a) < \infty$, and the distortion
function $d(x,\hat x):\mathcal{X}\times \mathcal{\hat X} \rightarrow
\mathbb{R}_{+}$ is defined such that for each $x\in \mathcal{X}$ there is an $\hat x\in
\mathcal{\hat X}$ satisfying $d(x,\hat x)=0$.
The rate-distortion-cost tuple $(R,D,C)$ is then said to be achievable if and only if, for all
$\varepsilon>0$, there exist an encoding
function $f$, an action function $g$ and a decoding function $h$, for all
sufficiently large $n\in\mathbb{N}$, satisfying 
the distortion constraint 
\begin{align}
  \mathbb{E}\left[ \sum_{i=1}^n d(X_i,\hat X_i) \right] \leq n (D+\varepsilon) 
  \label{eq:distortion}
\end{align} 
and the action cost constraint
\begin{align}
  \mathbb{E}\left[\sum_{i=1}^{n} \Delta(A_i)\right] \leq n (C+\varepsilon).
\label{eq:cost}
\end{align}
The rate-distortion-cost function, denoted as $R(D,C)$, is defined as the
infimum of all rates $R$ such that the tuple $(R,D,C)$ is achievable.

\subsection{Rate-Distortion-Cost Function}
The rate-distortion-cost function $R(D,C)$ was derived in \cite{permuter} and is
summarized below.
\begin{lemma}(\cite[Theorem 1]{permuter})
  The rate-distortion-cost function for the source coding problem with
  action-dependent side information is given as
  \begin{align}
    R(D,C) &= \min I(X;A) + I(X;U | Y, A),
    \label{eq:rate_dist_cost}  
  \end{align}
  \begin{align}
    P_{X,Y,A,U}(x,y,a,u) &= P_X(x)P_{U|X}(u|x)\mathbbm{1}_{\{\eta(u)=a\}} P_{Y|X,A}(y|x,a),
\label{eq:lem_joint_pmf}
  \end{align}
  and the minimization is over all pmfs
  $P_{U|X}$ and deterministic functions $\eta: \mathcal{U}\rightarrow
  \mathcal{A}$ under which the conditions
  \begin{align}
    \mathbb{E}[d(X,\hat X^{\text{opt}} (U,Y))] \leq D,
\label{eq:exp_dist}
  \end{align}
  and
  \begin{align}
    \mathbb{E}[\Delta(A)]\leq C
\label{eq:cost_eq_lem}
  \end{align}
hold. The function $\hat X^{\text{opt}}: \mathcal{U}\times \mathcal{Y}
\rightarrow \mathcal{\hat X}$ denotes the best estimate of $X$ given $U$ and
$Y$, i.e.,
\begin{align}
  \hat X^{\text{opt}}(u,y)=\arg \min_{\hat x\in \mathcal{\hat X}}
  \mathbb{E}[d(X,\hat x)|U=u,Y=y].
  \label{eq:def_xopt}
\end{align}
Moreover, the cardinality of the set $\mathcal{U}$ can be restricted as
$|\mathcal{U}|\leq |\mathcal{X}||\mathcal{A}|+2$.
\label{lem:permuter}
\end{lemma}

\subsection{Optimal Coding Strategy}
\label{sec:optimal_coding}
\begin{figure}[!t]
  \center
  \includegraphics[width=2.3in]{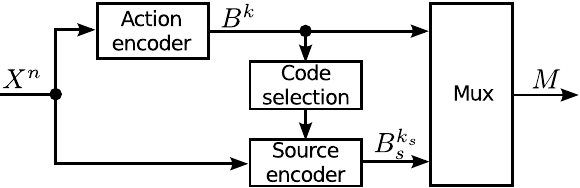}
  \caption{Optimal encoder for source coding problems with action-dependent side
    information.}
  \label{fig:optimal_code}
\end{figure}
The proof of achievability of the rate-distortion-cost function in
\cite{permuter} shows that an optimal encoder has the structure illustrated in
Fig.~\ref{fig:optimal_code} and consists of the following
two steps.
  \begin{itemize}
  \item \underline{Action Coding}: The source sequence $X^n$ is mapped to an action sequence
    $A^n$. The action sequence is selected from a codebook $\mathcal{C}_A$ of
    about $2^{n I(X;A)}$ codewords, each type approximately equal to $P_A$. The
    index $B^k$ identifies the selected codeword $A^n$, and hence consists of $k$,
    approximately equal to $nI(X;A)$, bits. The selection of $A^n$ is done with the
    aim of ensuring that $A^n$ and $X^n$ are jointly typical with respect to the
    joint pmf $P_{X,A}(x,a)=P_{A|X}(a|x)P_X(x)$.
  \item \underline{Source Coding}: Given the action sequence $A^n$, a source
    codebook is chosen out of a set of around $2^{nI(X;A)}$ codebooks, one for
    each codeword in $\mathcal{C}_A$. Each codeword $U^n$ in the selected source codebook has a
    joint type with $A^n$ close to $P_{A,U}$, and the number of codewords is about
    $2^{n I(X;U|A)}$. The source
    sequence is mapped to a sequence $U^n$ taken from the selected
    codebook with joint type $P_{A,U}$ and with the objective of ensuring
    that $X^n$, $A^n$ and $U^n$ are jointly typical with respect to the joint pmf
    $P_{X,A,U}(x,a,u)$. Each source codebook is divided into around $2^{nI(X;U|A,Y)}$
    subcodebooks, or bins, in order to leverage the side information at the
    receiver using Wyner-Ziv decoding.
  \end{itemize}
The message $M$ is given by the concatenation of the bits $B^k$ and $B^{k_s}_s$
and thus the overall rate of the action code and the source codes is given by
\eqref{eq:rate_dist_cost}. Upon receiving the message $M$ from the encoder, the
decoder first reconstructs the action sequence $A^n$. The action sequence is used to
measure the side information $Y^n$. As $A^n$ is known, the decoder also knows the
source codebook from which $U^n$ is selected, and $U^n$ is then recovered by
using Wyner-Ziv decoding based on the side information $Y^n$. In the end, the
final estimate $\hat X^n$ is obtained as $\hat X_i=\hat X^{\text{opt}}(U_i,
Y_i)$ for $i\in[1,n]$.

\section{Computation of the Rate-Distortion-Cost Function}
\label{sec:arimoto}
In this section, we first reformulate the problem in \eqref{eq:rate_dist_cost}
by introducing Shannon strategies. This result is then used to propose a
BA-type algorithm for the computation of the rate-distortion-cost function
\eqref{eq:rate_dist_cost}.

\subsection{Shannon Strategies}
We first observe that, from Lemma \ref{lem:permuter}, it is sufficient to
restrict the minimization to all joint distributions for which $A$ is a
deterministic function $A=\eta(U)$. Moreover, the final estimate of $\hat X$ in
\eqref{eq:def_xopt} is a function of both $U$ and $Y$. Based on these facts, we
define a Shannon strategy $T \in \mathcal{T}\subseteq
\mathcal{X}^{|\mathcal{Y}|}\times \mathcal{A}$ as a vector of cardinality
$|\mathcal{Y}|+1$, in which the first $|\mathcal{Y}|$ elements are indexed by
the elements in $\mathcal{Y}$ and $T(y)\in \mathcal{\hat X}$ for
$y\in\mathcal{Y}$, and the last element is denoted $\mathrm{a}(T)\in
\mathcal{A}$. We also define the disjoint
sets $\mathcal{T}^a=\{t\in \mathcal{T} : \mathrm{a}(t) = a\}$ for all actions
$a\in \mathcal{A}$. The rate-distortion-cost function \eqref{eq:rate_dist_cost}
can be restated in terms of the defined Shannon strategies as formalized in
the next proposition.
\begin{proposition}
  Let $T\in \mathcal{T}\subseteq \mathcal{X}^{|\mathcal{Y}|}\times \mathcal{A}$
  denote a Shannon strategy vector as defined above.
  The rate-distortion-cost function in
  \eqref{eq:rate_dist_cost} can be expressed as
  \begin{align}
    R(D,C) = \min I(X;\mathrm{a}(T)) + I(X;T | Y, \mathrm{a}(T)),
\label{eq:rdc_prop1}
  \end{align}
where the joint pmf $P_{X,Y,T}$ is of the form
\begin{align}
  &P_{X,Y,T}(x,y,t) = P_X(x)P_{T|X}(t|x)P_{Y|A,X}(y|\mathrm{a}(t),x),
\label{eq:joint_pxyt}
\end{align}
and the minimization is over all pmfs $P_{T|X}$ under the constraints
\begin{align}
  \mathbb{E}[\Delta(A)] = \sum_{t\in \mathcal{T},x\in\mathcal{X}}
  P_X(x) P_{T|X}(t|x)\Delta(\mathrm{a}(t)) \leq C
\label{eq:cost_ineq}
\end{align}
and
\begin{align}
  \mathbb{E}[d(X,T(Y))] = \sum_{t\in \mathcal{T}, x\in \mathcal{X}, y\in\mathcal{Y}} P_{X,Y,T}(x,y,t)d(t(y),x) \leq D.\label{eq:dist_ineq}
\end{align}
Moreover, the cardinality of the alphabet $\mathcal{T}$ can be restricted as
$|\mathcal{T}|\leq |\mathcal{X}| |\mathcal{A}|+2$.
\label{prop_stategies}
\end{proposition}
\begin{proof}
  Given an alphabet $\mathcal{U}$, a pmf $P_{U|X}$ and a function
  $\eta:\mathcal{U}\rightarrow \mathcal{A}$, the sum of the two mutual
  informations in \eqref{eq:rate_dist_cost} can be seen to be equal to the sum
  of the two mutual informations in \eqref{eq:rdc_prop1} and the average
  distortion and cost in \eqref{eq:exp_dist} and \eqref{eq:cost_eq_lem} to be
  equal to \eqref{eq:dist_ineq} and \eqref{eq:cost_ineq}, respectively, by
  defining $P_{T|X}$ as follows.
  For each $u\in \mathcal{U}$, define a strategy $t$ with
  $P_{T|X}(t|x)=P_{U|X}(u|x)$ such that $\mathrm{a}(t)=\eta(u)$ and $t(y) = \hat
  X^{\text{opt}}(u,y)$ for $y\in \mathcal{Y}$.
\end{proof}
\begin{rem}
  The characterization in Proposition \ref{prop_stategies} generalizes the
  formulation of the Wyner-Ziv rate-distortion function in terms of Shannon
  strategies given in \cite{dupuis}.
\end{rem}
The following lemma extends to the rate-distortion-cost function $R(D,C)$ some
well-known properties for the rate-distortion function (see, e.g. \cite{gray,
  cover:inf}). This will be useful in the next section when discussing the
computation of $R(D,C)$.
\begin{lemma} The following properties hold for the rate distortion
  cost-function $R(D,C)$:
  \begin{enumerate}
  \item $R(D,C)$ is non-increasing, convex and continuous for $D\in [0,\infty)$
    and $C\in [0,\infty)$.
  \item $R(D,C)$ is strictly decreasing in $D\in[0, D_{\text{max}}(C)]$ and
    $R(D_{\text{max}}(C),C)=0$, where
    \begin{align}
      &D_{\text{max}}(C) = \min_{P_T} \sum_{t\in\mathcal{T},x\in\mathcal{X},
        y\in\mathcal{Y}} P_{X,Y,T}(x,y,t)d(t(y),x),
      \label{eq:Dmax}
    \end{align}
    under the constraint
    \begin{align}
      \mathbb{E}[\Delta(\mathrm{a}(T))] = \sum_{t\in\mathcal{T}}\Delta(\mathrm{a}(t))P_T(t)\leq C.
      \label{eq:action_inde_cost}
    \end{align}
  \item For all $D\in [0,D_{\text{max}}(C)]$, the minimum in \eqref{eq:rdc_prop1} is attained when
    the distortion inequality \eqref{eq:dist_ineq} is satisfied with equality.
  \end{enumerate}
  \label{prop_CONVEXITY}
\end{lemma}
\begin{proof}
  The lemma is proved by the arguments in 
  \cite[Lemma 10.4.1]{cover:inf}. 
\end{proof}
\subsection{Computation of the Rate-Distortion-Cost Function}
\begin{algorithm}[tp]
  \caption{BA-type Algorithm for Computation of the Rate-Distortion-Cost Function}
  \label{alg:rdc}
  \begin{algorithmic}
    \STATE \textbf{input}: Lagrange multipliers
    $s\leq 0$ and $m\leq 0$.
    \STATE \textbf{output}: $R(D_{s,m},C_{s,m})$ with $C_{s,m}$ and $D_{s,m}$ as
    in \eqref{eq:Csm}-\eqref{eq:Dsm}.
    \STATE \textbf{initialize}: $P_{T|X}$
    \REPEAT
    \STATE Compute $Q_A$ as in \eqref{eq:Qa_min}.
    \STATE Compute $Q_{T,Y}$ as in \eqref{eq:Qty_min}.
    \STATE Minimize $F(P_{T|X},Q_{T,Y},Q_A)$ with respect to $P_{T|X}$ using
    Algorithm~\ref{alg:subgradient}.
    \UNTIL convergence
    \STATE $P_{T|X}^*\gets P_{T|X}$
  \end{algorithmic}
\end{algorithm}
  In order to derive a BA-type algorithm to solve the problem in
  \eqref{eq:rdc_prop1}, we introduce Lagrange multipliers $m$ for the cost
  constraint in \eqref{eq:cost_ineq} and $s$ for the distortion constraint
  \eqref{eq:dist_ineq}. The following proposition provides a parametric
  characterization of the rate-distortion-cost function in terms of the pair
  $(s,m)$.
\begin{proposition}
For each $s\leq 0$ and $m\leq 0$, define the rate-distortion-cost tuple
$(R_{s,m},D_{s,m},C_{s,m})$ via the following equations
\begin{align}
  R_{s,m} &=s D_{s,m}+m C_{s,m}\nonumber\\&\quad+\min_{P_{T|X}}\left\{ I(X;A)+I(X;T|Y, \mathrm{a}(T))
 - s\mathbb{E}\left[d(X,T(Y))\right]- m \mathbb{E}\left[
  \Delta(\mathrm{a}(T)) \right]\right\},
\label{eq:rate_dist_cost_def}\\
C_{s,m} &= \sum_{t\in \mathcal{T},x\in\mathcal{X}} P_X(x) P^*_{T|X}(t|x)
\Delta(\mathrm{a}(t))\label{eq:Csm},\\
D_{s,m} &= \sum_{t\in \mathcal{T}, x\in\mathcal{X}, y\in\mathcal{Y}} P_X(x)
P^*_{T|X}(t|x) P_{Y|X,A}(y|x,a) d(t(y),x),\label{eq:Dsm}
\end{align}
where $P^*_{T|X}$ denotes a minimizing pmf $P_{T|X}$ for the optimization
problem in \eqref{eq:rate_dist_cost_def}. Then, the following facts hold
\begin{enumerate}
\item The tuple $(R_{s,m},D_{s,m}, C_{s,m})$ lies on the rate-distortion-cost
  function, i.e.,
\begin{align}
R_{s,m} = R(D_{s,m},C_{s,m}).
\end{align}
\item Every point $(R,D,C)$ on the rate-distortion-cost function for
  $D\in[0,D_{\text{max}}(C)]$  can be written as
  \eqref{eq:rate_dist_cost_def}-\eqref{eq:Dsm} for  $s\leq 0$ and $m\leq
  0$;
\item The rate-distortion-cost function is given as
  \begin{align}&R(D,C) = \max_{\stackrel{s\leq 0}{m\leq 0}}
    \left(R_{s,m}+ s (D-D_{s,m})
+
      m(C-C_{s,m})\right).
  \end{align}
\end{enumerate}
\end{proposition}
\begin{proof}
  The proposition above follows by strong duality as
  guaranteed by Slater's condition \cite[Section 5.2.3]{boyd}, and can also be
  derived directly as in \cite{gray}.
\end{proof}
Given the proposition above, one can trace the rate-distortion-cost function
by solving problem \eqref{eq:rate_dist_cost_def} and using \eqref{eq:Csm} and
\eqref{eq:Dsm} for all $s\leq 0$ and $m\leq 0$. Inspired by the standard BA
approach, we now show that problem \eqref{eq:rate_dist_cost_def} can be solved
by using alternate optimization with respect to $P_{T|X}$ and
appropriately defined auxiliary pmfs $Q_{T,Y}$ and $Q_A$. To do this, we define
the function $F(\cdot )$ of $P_{T|X}$ and auxiliary pmfs $Q_{T,Y}$ and $Q_A$ as
in \eqref{eq:F}, 
\begin{align}
  &F(P_{T|X},Q_{T,Y},Q_A) = D_{KL}(P_{Y,A}||Q_A)- \sum_{x\in \mathcal{X}, y\in\mathcal{Y}, t\in \mathcal{T}} P_{X,Y,T}(x,y,t) \log
  P_{Y|X,A}(y|x,\mathrm{a}(t)) \nonumber\\
  &\quad+\sum_{x\in\mathcal{X}}P_X(x)D_{KL}(P_{Y,T|X}(\cdot,\cdot|x)||Q_{T,Y}) - s \sum_{t\in
      \mathcal{T}, x\in\mathcal{X}, y\in\mathcal{Y}} P_{X,Y,T}(x,y,t)  d(t(y),x)\nonumber\\
  &\quad- m \sum_{t\in
      \mathcal{T},x\in\mathcal{X}} \Delta(\mathrm{a}(t)) P_X(x)
    P_{T|X}(t|x),
  \label{eq:F}
\end{align}
where $D_{KL}(P||Q)$ denotes the Kullback-Leibler (KL) divergence\footnote{The
  Kullback-Leibler divergence  \cite{cover:inf} is defined as
  $D_{KL}(P||Q)=\sum_{i}P(i)\log_2 \frac{P(i)}{Q(i)}$ for pmfs $P$ and $Q$.}
and $P_{X,Y,T}$, $P_{Y,T|X}$ and $P_{Y,A}$ are calculated from the joint pmf \eqref{eq:joint_pxyt}.
We then have the following result.
\begin{proposition}
For any $s\leq 0$ and $m\leq 0$, we have
\begin{align}
  &R(D_{s,m},C_{s,m}) = sD_{s,m}+mC_{s,m} 
+ \min_{P_{T|X},Q_{T,Y},Q_A}
  F(P_{T|X},Q_{T,Y},Q_A),
\label{eq:prop_rdc}
\end{align}
with \eqref{eq:Csm}-\eqref{eq:Dsm}, where the distribution $P^*_{T|X}$ denotes a minimizing distribution in
\eqref{eq:prop_rdc}. Moreover, the function $F(P_{T|X}, Q_{T,Y},Q_A)$ is jointly
convex in the pmfs $P_{T|X}$, $Q_{T,Y}$ and $Q_A$.
\label{prop_rdc_alt_min}
\end{proposition}
\begin{proof}
  The proof technique for the first part is due to \cite{blahut}, and is
  based on showing that the pmf $Q_A$ minimizing $F(\cdot )$ for fixed $Q_{T,Y}$
  and $P_{T|X}$ is
\begin{align}
 Q_{A}(a)&= \sum_{x\in\mathcal{X},t\in \mathcal{T}^a}
    P_X(x)P_{T|X}(t|x)=P_A(a)\label{eq:Qa_min},
  \end{align}
and the pmf $Q_{T,Y}$ minimizing $F(\cdot)$ for fixed $Q_A$ and $P_{T|X}$ is
given by
\begin{align}
  &Q_{T,Y}(t,y)= \sum_{x\in \mathcal{X}}  P_X(x)P_{Y|X,A}(y|x,\mathrm{a}(t))
      P_{T|X}(t|x)=P_{T,Y}(t,y) .\label{eq:Qty_min}
\end{align}
The convexity of the function $F(\cdot)$ follows from the log-sum inequality
\cite{cover:inf}.
\end{proof}
Based on Proposition~\ref{prop_rdc_alt_min}, the proposed BA-type algorithm for
computation of the rate-distortion-cost function then consists of alternate
minimizing \eqref{eq:prop_rdc} with respect to $P_{T|X}$, $Q_{T,Y}$ and
$Q_A$. Due to the convexity of \eqref{eq:prop_rdc}, the algorithm is known to
converge to the optimal point similar to \cite{dupuis}. The proposed algorithm
is summarized in Table~Algorithm~\ref{alg:rdc}.  The step of minimizing
$F(P_{T|X}, Q_{T,Y}, Q_A)$ with respect to $P_{T|X}$ is discussed in the rest of this section.
\subsection{Minimizing $F$ over $P_{T|X}$}
To minimize the function $F(P_{T|X}, Q_{T,Y}, Q_A)$ with respect to $P_{T|X}$ for fixed
$Q_A$ and $Q_{T,Y}$, we add a Lagrange multipliers
$\lambda_x$ for each equality constraints $\sum_{t\in\mathcal{T}} P_{T|X}(t|x)=1$
with $x\in\mathcal{X}$, and resort to the KKT conditions as necessary and
sufficient conditions for optimality. 
This property of the KKT conditions
follows by strong duality due to the validity of Slater's conditions for the
problem \cite[Section 5.2.3]{boyd}.
We assume $P_X(x)>0$ without loss of
generality, since values of $x$ with $P_X(x)=0$ can be removed from the alphabet
$\mathcal{X}$.

By strong duality, we obtain the following optimization problem
\begin{align}
  &\min_{\stackrel{P_{T|X}\geq 0}{\sum_{t\in\mathcal{T}} P_{T|X}(t|x)=1}} F(P_{T|X},
  Q_A, Q_{T,Y}) =\nonumber\\
  &\qquad\qquad\max_{\{\lambda_x\}\in \mathbb{R}^{|\mathcal{X}|}}
    \min_{P_{T|X}}
    F(P_{T|X},Q_A,Q_{T,Y}) + \sum_{x\in\mathcal{X}} \lambda_x \left(
        \sum_{t\in\mathcal{T}}P_{T|X}(t|x)-1
      \right).
\label{eq:maxmin_min}
\end{align}
In the proposed approach, the outer maximization in \eqref{eq:maxmin_min} is
then performed using the standard subgradient method. The inner minimization is
instead performed by finding the stationary points of the function.  This leads
to the system of equalities $g_{a|x}(P_{A|X},\mu_x)=P_{A|X}(a|x)$ for
$a\in\mathcal{A}$ and $x\in\mathcal{X}$, with
\begin{align}
  g_{ a| x}(P_{A|X},\mu_x) &= P_{A|X}(a|x)^{\beta}\left(\frac{2^{\mu_x}\alpha_{ a, x}}{
      \prod_{y\in\mathcal{Y}} \left[\sum_{\tilde x\in\mathcal{X}}
        P_X(\tilde x)P_{Y|X,A}(y|\tilde x, a)P_{A|X}(a|\tilde x)\right]^{P_{Y|X,A}(y|
        x, a)}}\right)^{1-\beta},
    \label{eq:hax_pax}
\end{align}
where
\begin{align}
  \alpha_{t,x}&=Q_A(\mathrm{a}( t)) 2^{m \Delta(\mathrm{a}( t))}
\cdot 2^{\sum_{y\in\mathcal{Y}}
  P_{Y|X,A}(y| x,\mathrm{a}( t)) \left[s d( t(y), x)
    +\log
    Q_{T,Y}( t,y) \right]}\label{eq:alpha_tx},\\
\alpha_{a,x}&=\sum_{t\in\mathcal{T}^a} \alpha_{t,x}.\label{eq:alpha_ax}
\end{align}
  and $\beta\in(0,1)$ is a parameter of the algorithm (see Appendix A). 
\begin{proposition}
  The algorithm in Tables~Algorithm~\ref{alg:rdc} and
  Algorithm~\ref{alg:subgradient} converges to the rate-distortion-cost function $R(D_{s,m},
  C_{s,m})$ for all $s\leq 0$ and $m\leq 0$. 
\label{prop:alg_conv}
\end{proposition}
\begin{proof}
See Appendix~\ref{app:lem:alg_conv}.
\end{proof}

\begin{algorithm}[!t]
  \caption{Algorithm for Minimization of $F$ with respect to $P_{T|X}$}
  \label{alg:subgradient}
  \begin{algorithmic}
    \STATE \textbf{input}:  $Q_{T,Y}$ and $Q_A$.
    \STATE \textbf{output}: $P_{T|X}^*$.
    \STATE \textbf{parameters}: Subgradient weights
    $\theta_i=\frac{1}{i},i\in\mathbb{Z}_+$ and constant $\beta\in (0,1)$.
    \STATE \textbf{initialization}:  $i=0$; $\mu^{(0)}_x=1$ for $x\in\mathcal{X}$;
    $P^{(0)}_{A|X}(a|x)=\frac{1}{|\mathcal{T}|}$ for
    $t\in\mathcal{T},x\in\mathcal{X}$.
    \REPEAT
    \STATE Perform fixed-point iterations on the system
    $P_{A|X}(a|x)=g_{a|x}(P_{A|X}, \mu_x)$ for $a\in\mathcal{A}$ and $x\in\mathcal{X}$ with starting point $P^{(i)}_{A|X}$ until convergence to
    obtain $P^{(i+1)}_{A|X}$.
    \STATE Update the subgradients as\\
    $\mu^{(i+1)}_x=\mu_x^{(i)}+\frac{\theta_i}{P(x)}
    \left(1-\sum_{a\in\mathcal{A}}P_{A|X}^{(i+1)}(a|x)\right)$ for $x\in\mathcal{X}$.
    \STATE $i \gets i + 1$.
    \UNTIL convergence
    \STATE Compute $P_{T|X}^*( t| x) = \frac{\alpha_{t,x}}{\alpha_{\mathrm{a}(t),x}}P_{A|X}^{(i)}(\mathrm{a}(t)|x)$.
  \end{algorithmic}
\end{algorithm}

\section{Code Design}
\label{sec:code_design}

\begin{figure}[!t]
\centering
\subfigure[Encoder]{
  \includegraphics[width=2.5in]{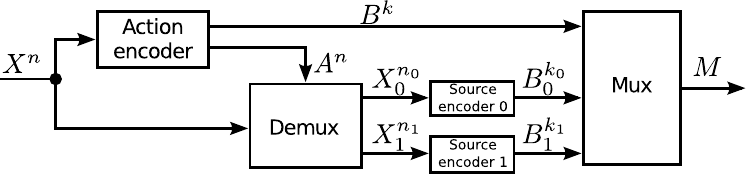}
	\label{fig:encoder}
}
\subfigure[Decoder]{
  \includegraphics[width=2.0in]{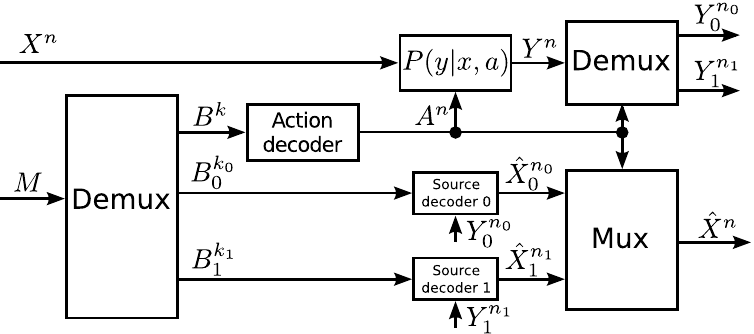}
	\label{fig:decoder}
}
\caption{Code design for source coding problems with action-dependent side
  information. The illustration is for $\mathcal{A}=\{0,1\}$.}
\label{fig:code_design}
\end{figure}
  In this section, we consider the design of specific encoders and decoders for
  the source coding problem with action-dependent side information.  The goal is
  to design codes that perform close the rate-distortion-cost function given in
  Lemma~\ref{lem:permuter} for some fixed pmf in \eqref{eq:lem_joint_pmf} (or
  equivalenty in Proposition~\ref{prop_stategies} for some fixed pmf
  $P_{X,Y,T}$).
\subsection{Achievability via Multiplexing}
As explained in Section~\ref{sec:optimal_coding}, the achievability proof in
\cite{permuter} is based on an action codebook $\mathcal{C}_A$ for the action
sequences $A^n$ of about $2^{n I(X;A)}$ codewords and $2^{nI(X;A)}$ source
codebooks of about $2^{n I(X;U|A)}$ codewords for the sequences $U^n$, where
each source codebook corresponds to an action sequence $A^n$.  We also recall
that binning is performed on the source codebooks in order to reduce the rate.

Here, we first observe that the code design can be simplified without loss of
optimality by using the encoder and decoder structures in
Fig. \ref{fig:code_design}. Accordingly, as in \cite{permuter}, the action
encoder selects the action sequence $A^n$, and the corresponding index $B^k$,
from the codebook $\mathcal{C}_A$ to the decoder, where $k=\lceil n I(X;A)
\rceil$. However, rather than using $2^{nI(X;A)}$ source codebooks, we utilize
only $|\mathcal{A}|$ source codebooks $\mathcal{C}_{s,a}$,
$a\in\mathcal{A}$. Specifically, the source codebook $\mathcal{C}_{s,a}$ has
about $2^{n P_A(a) I(X;U|A=a)}$ codewords, and each codeword in codebook
$\mathcal{C}_{s,a}$ has a length of $n_a=\lceil n (P_A(a)+\varepsilon) \rceil$
symbols for some $\varepsilon>0$.

To elaborate, as seen in Fig.~\ref{fig:encoder}, after action encoding, which
takes place as in \cite{permuter}, the source $X^n$
is demultiplixed into $|\mathcal{A}|$ subsequences, such that the $a$-th
subsequence $X_a^{n_a}$ contains all symbols $X_i$ for which $A_i=a$. Therefore,
for sufficiently large $n$, by the law of large numbers, the number of symbols
in $X_a^{n_a}$ is less than $n_a$ with high probability. Appropriate padding is
then used to make the length of the sequence exactly $n_a$ symbols. The $a$-th
subsequence $X_a^{n_a}$ is then compressed using the codebook
$\mathcal{C}_{s,a}$ with the objective of ensuring that $X_a^{n_a}$ and
$U_a^{n_a}$ are jointly typical with respect to the pmf $P_{X,U|A}(\cdot, \cdot |
a)$. Binning is performed on each source codebook so that the number of bins is
$2^{n P_A(a) I(X;U|Y,A=a)}$. 
The bin index $B_a^{k_a}$ of $U_a^{n_a}$ is thus of
$k_a=\lceil n P_A(a) I(X;U|Y,A=a) \rceil$ bits.
Overall, the rate of the message
$M$, consisting of the indices $B^k$ for the action code and $B_{s,a}^{k_a}$ for
the source codes with $a\in\mathcal{A}$, is
$I(X;A)+\sum_{a\in\mathcal{A}} P_A(a)I(X;U|Y,A=a)=I(X;A)+I(X;U|A,Y)$ as desired.

At the decoder, as seen in Fig.~\ref{fig:decoder}, the action sequence $A^n$ is
reconstructed and is used to measure the side information
$Y^n$. The side information $Y^n$ is demultiplexed into $|\mathcal{A}|$
subsequences, such that the $a$-th subsequence $Y_a^{n_a}$ contains all symbols
$Y_i$ for which $A_i=a$. Each of the subsequences $U_a^{n_a}$ are then
reconstructed by using Wyner-Ziv decoding based on the message bits $B_a^{k_a}$
and the side information $Y_a^{n_a}$, and the reconstructed source subsequences
$\hat X_{a,i}$ are obtained as $\hat X_{a,i} = \hat X^{\text{opt}}(U_{a,i},
Y_{a,i})$ for $i\in[1,n_a]$, where $\hat X_{a,i}$ denotes the $i$-th symbol of
the sequence $X_a^{n_a}$. Finally, the source reconstruction $\hat X^n$ is
obtained by multiplexing the subsequences $\hat X_{a}^{n_a}$ for
$a\in\mathcal{A}$.
\begin{rem}
  The proposed code structure also applies to the classical successive
  refinement problem \cite{successive_refinement} and can be used to simplify the code design
  proposed in \cite{ldgm}.
\end{rem}

\subsection{The Action Code}
Based on the encoder structure in Fig.~\ref{fig:encoder}, we discuss the
specific design of the action encoder. The action code
$\mathcal{C}_A$ has to ensure that the codewords $A^n$ approximately have the
type $P_A$, and the action encoder must obtain a codeword $A^n$ that is jointly
typical with respect to the joint pmf $P_{X,A}$. These conditions are satisfied by
optimal source codes \cite{polar_sc}. Optimal source codes can be designed using
LDGM codes or polar codes as shown in \cite{nonuniform_ldgm} and
\cite{polar_sc}, respectively. Here, we adopt LDGM codes as proposed in
\cite{nonuniform_ldgm, bip}. Specifically, in the following, we define an encoder
based on message passing. This uses ideas from \cite{nonuniform_ldgm} to handle
the general alphabet and pmf $P_A$, and from \cite{bip} to implement message
passing and decimation. The key difference with respect to \cite{bip} is that
there the goal of the encoder is to minimize the Hamming distance, while
the aim in this paper is to find an action sequence that is jointly typical with the source.

We use the code described by the factor graph in Fig.~\ref{fig:pax_code}. The
bottom section of the graph is a LDGM code (see,
e.g. \cite{nonuniform_ldgm}). The sequence $B^k$ denotes the message bits with
$k=\lceil n I(X;A)\rceil$ and $\{g_{\kappa,l}:\kappa\in[1,d],l\in[1,n]\}$
denote the check variables of the LDGM code, where the choice of $d$ is explained
later. The objective of the mappings $\psi_l: \{0,1\}^d \times
\mathcal{A}\rightarrow \{0,1\}$ for $l\in[1,n]$ is to ensure that the
types of the codewords, or action variables, are approximately equal to
$P_A$ \cite{nonuniform_ldgm}. Specifically, each mapping $\psi_l$ applies to
the subset of check variables $\{g_{\kappa,l}\}_{\kappa \in[1,d]}$ and to the
symbol $a_l$ and is defined in terms of a mapping $\phi: \{0,1\}^d \rightarrow
\mathcal{A}$ as
\begin{align}
  \psi_l(\{g_{\kappa,l}\}_{\kappa\in[1,d]},a) = \mathbbm{1}_{\{ \phi(\{g_{\kappa,l}\}_{\kappa\in[1,d]})=a \}}.
\end{align}
Following \cite{nonuniform_ldgm}, the value of $d\in\mathbb{Z}_+$ is chosen
such that there are integers $\nu_a$ for $a\in\mathcal{A}$ satisfying 
\begin{align}
\sum_{a\in\mathcal{A}}\nu_a=2^d \qquad \text{and} \qquad P_A(a)\approx\frac{\nu_a}{2^d}.
\end{align}
The mapping $\phi$ is then arbitrarily chosen such that exactly $v_a$ of the
$2^d$ binary sequences $\{g_{\kappa,l}\}_{\kappa\in [1,d]}$ map to $a$.

Given the source sequence $X^n$, the encoder runs the sum-product algorithm with
decimation as in \cite{bip} in order to obtain the message bits $B^k$, and hence
the action sequence $A^n$ (see \cite{polar_sc} for a discussion of the role of
decimation in source coding problems).

\begin{figure}[tpb]
  \center
  \includegraphics[width=3.5in]{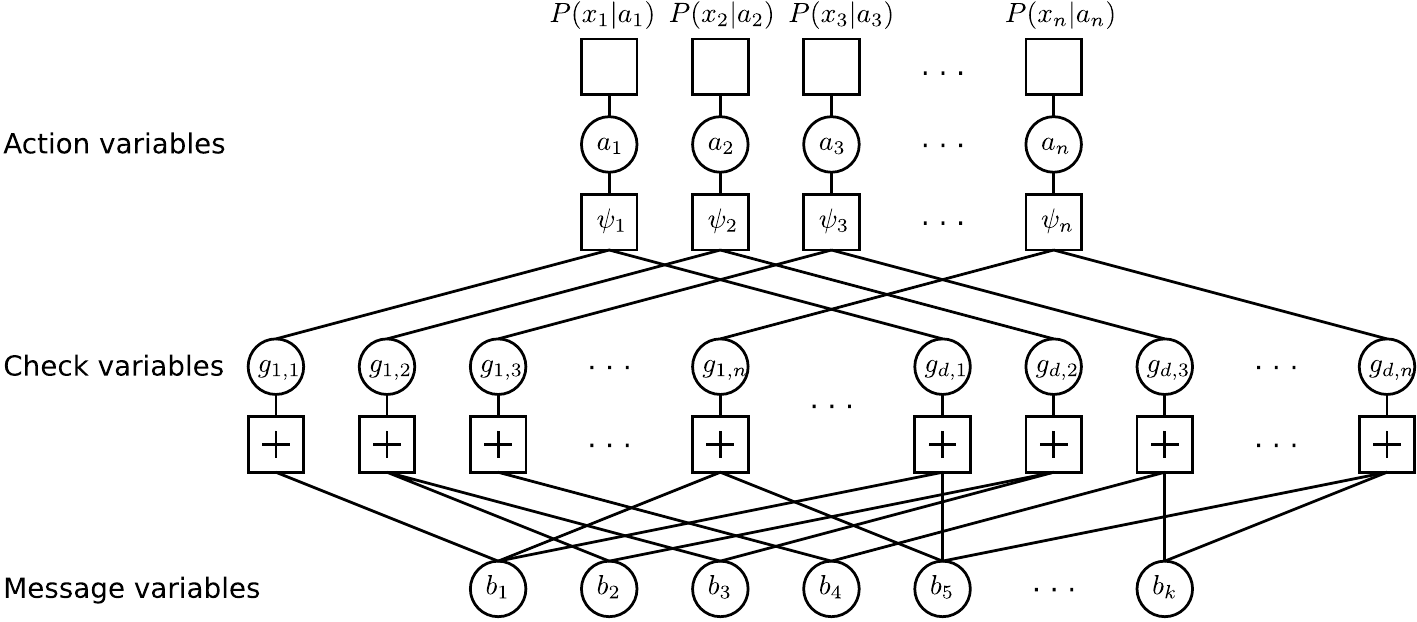}
  \caption{Factor graph defining the action encoder.}
  \label{fig:pax_code}
\end{figure}

\subsection{The Source Codes}
Based on the proposed encoder structure in Fig.~\ref{fig:encoder}, the design of
each source code $C_{s,a}$ for $a\in\mathcal{A}$ is equivalent to optimal codes
for classical Wyner-Ziv 
problems.

In the special case where $\mathcal{\hat X}=\{0,1\}$, and the distortion metric
is Hamming, the coding problem reduces to the binary Wyner-Ziv problem with
Hamming distortion which was studied in \cite{compound_codes,polar_sc}.

\section{Numerical Examples}
\label{sec:examples}
To exemplify the problems of interest and to demonstrate the tools developed in
this paper, we consider the source coding problem with action-dependent side
information depicted in Fig.~\ref{fig:simple_setup} and described in the
following.
Let $X\in \mathcal{X} = [1,K+1]$ be a random
variable with pmf 
\begin{align}
  P_X(x) = \left\{
    \begin{array}{ll} 
      \frac{1-q}{K} & \text{if}\quad x\in[1,K] \\
      \mathrm{q} & \text{if}\quad x=K+1
    \end{array}
  \right.,
\end{align}
for $q\in [0,1]$. The letters $1,\ldots,K$ denote source outcomes that are
relevant for the decoder, and thus should ideally be distinguishable by the
latter, while the letter $x=K+1$ represents a source outcome that is irrelevant
for the decoder. 
Examples where this situation arises includes monitoring
systems in which the decoder wishes to recover the values of a physical quantity
only when above, or below, a certain pre-determined threshold. To account for
this requirement, the distortion function is given by
\begin{align}
  d(x,\hat x) =\mathbbm{1}_{\{x\not=\hat x \text{ and } x\in[1,K]\}}
\label{eq:ex_dist}
\end{align}
i.e., the decoder is only penalized if it makes an error when $x$ is a relevant
letter.

At each time $i$, the decoder can choose an action $A_i\in \{0,1\}$, such that,
if $A_i = 0$, the side information is given by $Y_i = \mathrm{e}$, where
$\mathrm{e}$ denotes an erasure symbol, and if $A_i = 1$, the side information
is given by $Y_i = \tilde Y_i$, where $\tilde Y_i$ is the output of an erasure
channel in which $\mathcal{\tilde Y}=\mathcal{X}\cup \{\mathrm{e}\}$ and
\begin{align}
  P_{\tilde Y|X}(\tilde y|x) = \left\{
    \begin{array}{ll}
      p &  \text{for } \tilde y = \mathrm{e} \\
      1-p & \text{for }\tilde y=x\\
      0 &   \text{otherwise}
    \end{array}
  \right.,
\end{align}
where $p\in (0,1)$ is the erasure probability. The action cost function
$\Delta(\cdot)$ is given by $\Delta(a) = \mathbbm{1}_{\{a=1\}}$, which implies
that the cost constraint with $0 \leq C \leq 1$ enforces that no more than $nC$
samples of the side information $\tilde Y^n$ can be measured by the receiver.
\begin{figure}[!t]
  \center
  \includegraphics[width=2.5in]{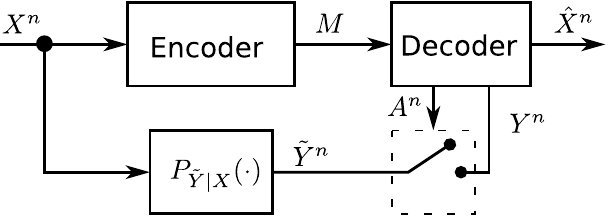}
  \caption{The action-dependent source coding problem. }
  \label{fig:simple_setup}
\end{figure}

\subsection{Computation of the Rate-Distortion-Cost Function}
We apply the proposed BA-type algorithm to the described scenario in order to
compute the rate-distortion-cost function. For reference, we also consider the
simplified strategy, in which the actions are chosen independently of the
message $M$. We refer to the optimal approach discussed thus far as ``adaptive
actions'', while labeling as ``non-adaptive actions'' the simplified class of
strategies in which the actions are selected independently of the encoder's
message (see \cite{permuter}).  The performance with non-adaptive actions can be
obtained from Proposition~\ref{prop_stategies} by imposing that $A$ and $X$ are
independent.


Fig.~\ref{fig:comp_rdc} shows $R(D,C)$ for $K=4$, $q=\frac{1}{2}$ and
$p\in\{0,0.1\}$ with both adaptive actions and non-adaptive actions. We see that
for the given scenario, we achieve significant gains using adaptive actions in
comparison to non-adaptive actions. Moreover, the effect of the erasures
decreases as the action cost decreases due to the reduced availability of the
side information at the decoder.

\begin{figure}[!t]
  \center
  \includegraphics[width=3in]{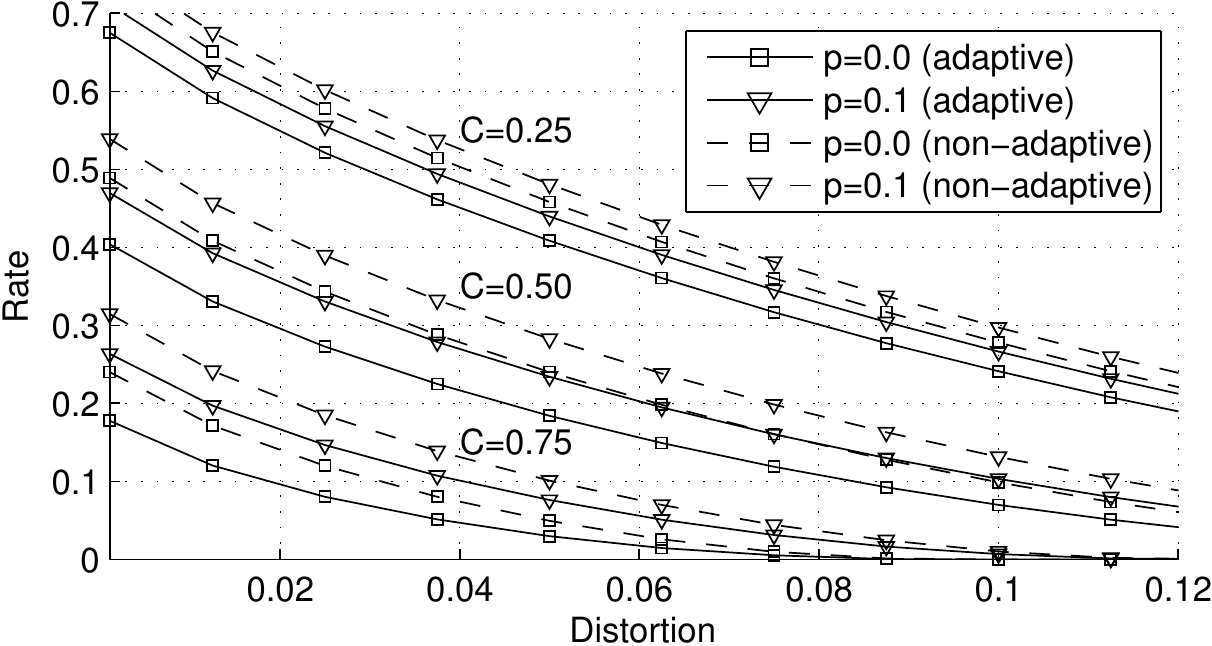}
  \caption{Computed rate-distortion-cost function $R(D,C)$ for $K=4$, erasure
    probability $p\in \{0.0,0.1\}$ and $q=\frac{1}{2}$.}
  \label{fig:comp_rdc}
\end{figure}
\subsection{Code Design}
  We now turn to the issue of code design for the scenario. We consider the
  case in which $p=0$, so that, the measured side information is noiseless and
  we adopt the code design proposed in Section~\ref{sec:code_design}.
We start with some analytical considerations of the rate-distortion-cost
function that will be useful for designing the codes. 
By symmetry,
the pmf $P_{A|X}$ can be written as
\begin{align}
   P_{A|X}(a|x) = \left\{
    \begin{array}{ll} 
      \frac{C-q \gamma}{1-q} & \text{if } a=1 \wedge x\in[1,K] \\
      \frac{1-q-C+q\gamma}{1-q} & \text{if } a=0 \wedge x\in[1,K]\\
      \gamma &\text{if }a=1 \wedge x=K+1\\
      1-\gamma &\text{if } a=0 \wedge x=K+1
    \end{array}
  \right.,
\label{eq:pax_opt}
\end{align}
where $\gamma\in\left[0,\min\left(1,\frac{C}{q}\right)\right]$ is a parameter to
be determined. The mutual information $I(X;A)$ can thus be computed in terms of
$P_{A|X}$ and $P_X$,
and the rate-distortion-cost function in \eqref{eq:rate_dist_cost} is then obtained
via the following optimization problem
\begin{align}
&R(D,C)=\min_{\gamma\in\left[0,\min\left(1,\frac{C}{q}\right)\right]} I(X;A)
+ (1-C)\bar R\left( \frac{D}{1-C}, P_{X|A=0}\right) ,
\label{eq:ex_rdc}
\end{align}
where $\bar R(D,P_X)$ is the classical rate-distortion function of a memoryless
source with pmf $P_X$. Note that we have used the fact that $I(X;U|Y,A=1)=0$
 since $Y=X$ for $A=1$. 

  From \eqref{eq:ex_rdc}, it is seen that we only need to design an action code
  and the source code $\mathcal{C}_{s,0}$, where the latter is a classical
  rate-distortion code. For the action code, we use the approach proposed in
  Section~\ref{sec:code_design} and for the source code we use the related LDGM
  scheme proposed in \cite{nonuniform_ldgm}.

We consider the case where $q=\frac{1}{2}$, $K=4$, which yields $d=2$ for both
the action code $\mathcal{C}_A$ and the source code $\mathcal{C}_{s,0}$. We fix
a blocklength of $n=\SI{10000}{}$, yielding LDGM codes of blocklength,
\SI{20000}{}. Each point is averaged over $50$ source realizations and LDGM
codes. For both codes, we use the sum-product algorithm with decimation in
\cite{bip}. As in \cite{bip}, we use damping after \SI{30}{} iterations and the
maximum number of iterations is set to $100$. Nodes are decimated if their
log-likelihood ratios are larger than $2$. Suitable irregular degree
distributions optimized for the AWGN channel are obtained from \cite{lopt}. The
results are shown in Fig. \ref{fig:numerical}. 
\begin{figure}[t]
  \center
  \includegraphics[width=3.0in]{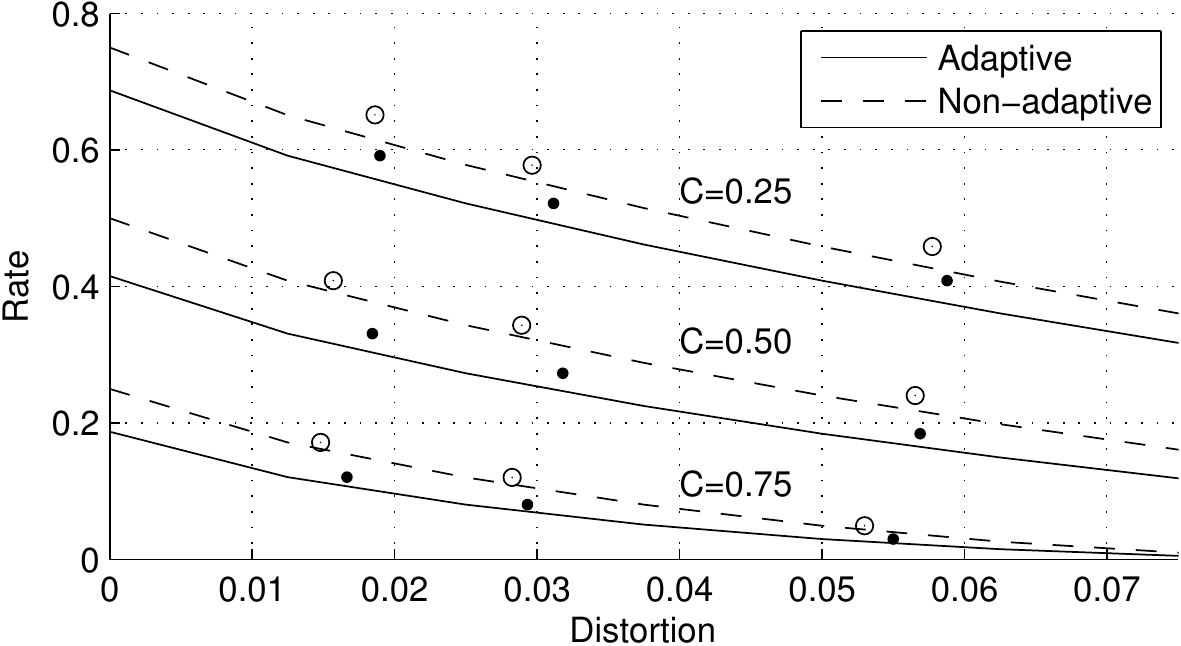}
  \caption{Rate-distortion-cost function (lines) compared to the performance of
    the proposed code design (markers) with both  adaptive and non-adaptive
    actions.}
  \label{fig:numerical}
\end{figure}
It is seen that the resulting distortions are close the lower bounds for both
the adaptive and non-adaptive actions strategies. Moreover, the theoretical
gains of the adaptive action strategy versus the non-adaptive one are confirmed
by the practical implementation.

\section{Conclusion}
In this paper, we have considered computation of the rate-distortion-code
function and code design for source coding problems with action-dependent side
information. We have formulated the problem using Shannon strategies and
proposed a BA-type algorithm that efficiently computes the rate-distortion
function. Convergence of this algorithm was proved. Moreover, we proposed a code
design based on multiplexing that was shown, via numerical results, to perform
close to the rate-distortion bound.


%

\appendices
\section{Proof for Lemma \ref{prop:alg_conv}}
\label{app:lem:alg_conv}
The BA-type algorithm detailed in Tables~Algorithm~\ref{alg:rdc} and
Algorithm~\ref{alg:subgradient} is based on alternatively optimizing $F(\cdot)$
in \eqref{eq:F} with respect to $P_{T|X}$, $Q_A$ and $Q_{T,Y}$. Given the
convexity of this function, shown in Proposition~\ref{prop_rdc_alt_min}, this
procedure is known to converge \cite{pol87}. The optimization with respect $Q_A$
for fixed $P_{T|X}$ and $Q_A$ and with respect to $Q_{T,Y}$ for fixed $P_{T|X}$
and $Q_{T,Y}$ are performed as in the proof of
Proposition~\ref{prop_rdc_alt_min}. Therefore, the proof is concluded once it is
demonstrated that the procedure of Table~Algorithm~\ref{alg:subgradient}
converges to an optimum $P_{T|X}$ for fixed $Q_A$ and $Q_{T,Y}$. This is
discussed next. The procedure in Table~Algorithm~\ref{alg:subgradient} for the
optimization with respect to $P_{T|X}$ for fixed $Q_A$ and $Q_{T,Y}$ is based on
the dual minimization \eqref{eq:maxmin_min} via an outer loop that performs
subgradient iterations and an inner loop that performs fixed-point iterations to
obtain a stationary point of the Lagrangian function \eqref{eq:lagrangian} (see
below). We first show that this nested loop procedure obtains an optimal dual
solution $P_{T|X}$ of the dual problem and then argue that this is also a
solution for the original primal problem. 

Convegence of the outer loop follows immediately by the well-known properties of
the subgradient approach for weights that are selected as $\Theta_i=\frac{1}{i}$
\cite{pol87}. Note that the constraints
$1-\sum_{a\in\mathcal{A}}P^{(i)}(a|x)$ for $x\in\mathcal{X}$ are the
subgradients with respect to $\lambda_{x}$ of the dual function given by the
minimization in \eqref{eq:maxmin_min} \cite{bertsekas:supp}. Therefore, by
defining $\mu_x=-\frac{\lambda_x}{P(x)}+2$, the updates of the variables
$\mu_x^{(i)}$ in Table~Algorithm~\ref{alg:subgradient} can be seen to correspond
to the classical subgradient updates. Given the known convergence properties of
the subgradient method with the weights as in
Table~Algorithm~\ref{alg:subgradient}, the outer maximation converges
\cite{pol87}.

Next, we need to show that we can solve the inner minimization in
\eqref{eq:maxmin_min} by using the fixed-point iterations in
\eqref{eq:fixed_point_iter} (see below). It is
first shown that we can solve the minimization problem by solving a system of
stationarity equations for $P(a|x),a\in\mathcal{A},x\in\mathcal{X}$. Then, we
conclude the proof using Banach fixed-point theorem \cite{intro:numerical}.

The Lagrangian to be minimized is given by (cf. \eqref{eq:maxmin_min})
\begin{align}
  &\mathcal{L}(P_{T|X},\{\lambda_x\})=F(P_{T|X},Q_A,Q_{T,Y})+\sum_{x\in\mathcal{X}}\lambda_x\left(\sum_{t\in\mathcal{T}}P_{T|X}(t|x)-1\right).
\label{eq:lagrangian}
\end{align}
It is noted that the function $\mathcal{L}$ is coercive in $P_{T|X}$, and hence from
Weierstrass theorem \cite{bertsekas} a minimizer of $\mathcal{L}$
exists. The minimizer must be a stationary point, i.e., it must satisfy the KKT
conditions \cite[Section 5.5.3]{boyd}.  We obtain the following stationarity
conditions by differentiating \eqref{eq:lagrangian} with respect to $P(t| x)$
and equating to zero, leading to
\begin{align}
&\log P(t| x)+ \sum_{y\in\mathcal{Y}}P(y| x,\mathrm{a}(
  t))\log \left[ P(y,\mathrm{a}(t)\right] =\nonumber\\
  &\quad m \Delta(\mathrm{a}( t)) +\sum_{y\in\mathcal{Y}}
  P(y| x,\mathrm{a}( t)) \left[s d( t(y),\tilde x)
    +\log
    Q( t,y)+\log Q(\mathrm{a}( t))\right]+\mu_x,\\
&\log P(t| x)+ \sum_{y\in\mathcal{Y}}P(y| x,\mathrm{a}(
  t))\log \left[ P(y,\mathrm{a}(t)\right] = \log \alpha_{t,x}+\mu_x\label{eq:Feq}
\end{align}
where $\alpha_{t,x}$ is given in \eqref{eq:alpha_tx} and $P(y,a)$ is calculated
from the joint pmf in \eqref{eq:joint_pxyt}. We can then rewrite \eqref{eq:Feq}
by applying the exponential function to both sides and solving for $P(t|x)$
\begin{align}
  &P( t| x)=\frac{2^{\mu_{x}} \alpha_{ t, x}}{\prod_{y\in\mathcal{Y}} \left[
      \sum_{\tilde x\in\mathcal{X}}
P(\tilde x)P(\mathrm{a}( t)|\tilde x)P(y|\tilde x,\mathrm{a}( t))
\right]^{P(y| x,\mathrm{a}( t))}}
\label{eq:ptx_eq}
\end{align}
where $\alpha_{t,x}$ is given in \eqref{eq:alpha_tx}. Note that the right-hand
side only depends on $P_{T|X}$ through $P_{A|X}$, and hence by computing
$P(a|x)$ for $a\in\mathcal{A}$ and $x\in\mathcal{X}$, $P(t|x)$ can be
calculated. By summing \eqref{eq:ptx_eq} over $ t\in\mathcal{T}^{ a}$, we obtain
\begin{align}
  P( a| x)&=\frac{2^{\mu_{x}}
  \alpha_{ a, x}}{\prod_{y\in\mathcal{Y}} \left[
    \sum_{\tilde x\in\mathcal{X}} P(\tilde x)P(y|\tilde x, a)P(a|\tilde x)\right]^{P(y|
    x, a)}},
\label{eq:pax_sys_proof}
\end{align}
where $\alpha_{a,x}$ is given in \eqref{eq:alpha_ax}. Given $\{\mu_x\}$, the
equalities in \eqref{eq:pax_sys_proof} for $a\in\mathcal{A}$ and
$x\in\mathcal{X}$ form a system of $|\mathcal{A}| |\mathcal{X}|$ nonlinear
equation with $|\mathcal{A}| |\mathcal{X}|$ unknowns, namely the $|\mathcal{A}|
|\mathcal{X}|$ values $P(a|x)$ for $a\in\mathcal{A}$ and $x\in\mathcal{X}$. By
solving for $P(a|x)$, we can compute $P(t|x)$ as in \eqref{eq:ptx_eq}. Note that
the constants $\alpha_{ a, x}$ are sums of exponential functions, and hence $P(
a| x)$ in \eqref{eq:pax_sys_proof} are strictly positive for $ a\in\mathcal{A}$
and $ x\in\mathcal{X}$. Now, define
\begin{align}
h_{a|x}(P_{A|X}, \mu_x)&=\frac{2^{\mu_x}\alpha_{ a, x}}{
      \prod_{y\in\mathcal{Y}} \left[\sum_{\tilde x\in\mathcal{X}}
      P_{X, A, Y}(\tilde x,a,y)\right]^{P_{Y|X,A}(y|
        x, a)}},\label{eq:h_ax}\\
 H_{a|x}(\mathbf{q}, \mu_x)&= \log h(2^{\mathbf{q}},\mu_x), \label{eq:H_ax}\\
\text{and } G_{a|x}(\mathbf{q}, \mu_x)&=\log g_{a|x}(2^{\mathbf{q}},\mu_x)\nonumber \\
&= \beta \mathbf{q}+(1-\beta) H_{a|x}(\mathbf{q},\mu_x)\label{eq:G_ax}
\end{align}
where $\mathbf{q}\in\mathcal{R}^{|\mathcal{A}||\mathcal{X}|}$ and
$2^{\mathbf{q}}\in\mathcal{R}_+^{|\mathcal{A}||\mathcal{X}|}$ are the vectors
corresponding to the elements $q_{a|x}=\log P(a|x)$ and $P(a|x)$,
respectively, and $\beta\in(0,1)$. Moreover, let
$\mathbf{G}(\mathbf{q},\{\mu_x\})\in\mathcal{R}^{|\mathcal{A}| |\mathcal{X}|}$
denote the vectors collecting the functions $G_{a|x}$ for
$a\in\mathcal{A},x\in\mathcal{X}$. With these definitions it is now evident that
\eqref{eq:pax_sys_proof} is equivalent to the following equation
\begin{align}
  q_{a|x} = H_{ a| x}(\mathbf{q},\{\mu_x\}).
\label{eq:sys_nonlin}
\end{align}

We now show that the fixed-point iteration of the form
\begin{align}
\mathbf{q}^{(k+1)}=\mathbf{G}(\mathbf{q}^{(k)},\{\mu_x\}).
\label{eq:fixed_point_iter}
\end{align}
converges towards a fixed-point $\mathbf{q}^*$, which is a unique fixed-point of
\eqref{eq:sys_nonlin} for any $\beta\in(0,1)$.

Recall that the existence of a fixed-point $\mathbf{q}^*$ is guaranteed by the
necessity of the KKT conditions and by Weierstrass theorem. In the following, we
apply Banach fixed-point theorem. To this end, we have to demonstrate that there
is a closed subset $\Omega\in \mathbb{R}^{|\mathcal{A}||\mathcal{X}|}$, such
that the vector function $\mathbf{G}$ maps from vectors $\mathbf{q}\in\Omega$ into
$\Omega$, and is a contraction in $\Omega$. By the existence of a fixed-point $\mathbf{q}^*$,
we define the subset $\Omega$ as the closed ball
\begin{align}
  \Omega = B_{r}(\mathbf{q}^*) = \left\{ \mathbf{q}\in \mathbb{R}^{|\mathcal{A}|
      |\mathcal{X}|} \big| \vectornorm{\mathbf{q} - \mathbf{q}^*}_{\infty} \leq r
  \right\},
\end{align}
for some $r> \vectornorm{\mathbf{q}^{(0)} - \mathbf{q}^*}_{\infty}$.  In order
to show that $\mathbf{G}$ maps from $\Omega$ into $\Omega$ and is a contraction,
we compute the partial derivatives of $H_{a|x}(\mathbf{q})$ and
$G_{a|x}(\mathbf{q})$ as following
  \begin{align}
    \frac{\partial H_{\tilde a|\tilde
        x}(\mathbf{q})}{\partial q_{a'|x'}} &=-\mathbbm{1}_{\{\tilde
    a=a'\}} \sum_{y\in\mathcal{Y}}P(y|\tilde x,\tilde
    a)\frac{ P(x')P(y|x',a')2^{q_{a'|x'}}  }{\sum_{x\in\mathcal{X}}
      P(x)P(y|x,\tilde a)2^{q_{\tilde a|x}}}\\
\text{and } \frac{\partial G_{\tilde a|\tilde
        x}(\mathbf{q})}{\partial q_{a'|x'}}&= \beta \mathbbm{1}_{\{\tilde
    a=a' \text{ and } \tilde x=x'\}} + (1-\beta) \frac{\partial H_{\tilde a|\tilde
        x}(\mathbf{q})}{\partial q_{a'|x'}}.
  \end{align}
  It is clear that the derivative $\frac{\partial H_{\tilde a|\tilde
      x}(\mathbf{q})}{\partial q_{\tilde a|\tilde x}}$ is strictly negative for
  $\mathbf{q}\in \mathbb{R}^{|\mathcal{A}| |\mathcal{X}|}$ since $P(x)>0$, and
  it can be seen that
  \begin{align}
    \sum_{a'\in\mathcal{A}, x'\in\mathcal{X}} \frac{\partial H_{\tilde a|\tilde
        x}(\mathbf{q})}{\partial q_{a'|x'}} = -1.
  \end{align}
  Therefore, for $\beta\in(0,1)$, we must have that
  \begin{align}
    \sum_{a'\in\mathcal{A}, x'\in\mathcal{X}} \left|\frac{\partial G_{\tilde a|\tilde
        x}(\mathbf{q})}{\partial q_{a'|x'}}\right| < 1.
  \end{align}
  It follows that we can bound the $l_{\infty}$-norm of the Jacobian for
  $\mathbf{G}(\mathbf{q})$, $J_{\mathbf{G}}(\mathbf{q})$, as
  \begin{align}
    \vectornorm{J_{\mathbf{G}}(\mathbf{q})}_{\infty} < 1.
  \end{align}
  By the definition of the $l_{\infty}$-norm and by the mean value theorem
  \cite{intro:numerical}, there exist values $\tilde a \in\mathcal{A}$, $\tilde
  x\in\mathcal{X}$ and $\zeta\in (0,1)$ such that
\begin{subequations}
  \begin{align}
    \vectornorm{\mathbf{G}(\mathbf{q}^1)-\mathbf{G}(\mathbf{q}^2)}_{\infty}
    &= |G_{\tilde a|\tilde x}(\mathbf{q}^1)-G_{\tilde a|\tilde
      x}(\mathbf{q}^2)|\\
    &\leq \vectornorm{\mathbf{q}^1-\mathbf{q}^2}_{\infty}  \sum_{a\in\mathcal{A}, x\in\mathcal{X}}
    \left| \frac{\partial G_{\tilde a|\tilde
          x}}{\partial q_{a|x}} ( \zeta \mathbf{q}^1 + (1-\zeta)
      \mathbf{q}^2)\right|\\
    &\leq \vectornorm{\mathbf{q}^1-\mathbf{q}^2}_{\infty} \max_{\mathbf{q}\in \Omega}
    \vectornorm{J_{\mathbf{G}}(\mathbf{q})}_{\infty}\\
    & \leq K \vectornorm{\mathbf{q}^1 - \mathbf{q}^2}_{\infty} 
  \end{align}
    \label{eq:mean_value}
\end{subequations}
for $\mathbf{q}^1,\mathbf{q}^2\in \Omega$, where the last inequality follows by
the fact that $\vectornorm{J_{\mathbf{G}}(\mathbf{q})}_{\infty}$ must attain a
maximum value $K<1$ when $\mathbf{q}\in \Omega$, since $\Omega$ is closed and
bounded, by Weierstrass theorem. The chain of inequalities in
\eqref{eq:mean_value} demonstrates that $\mathbf{G}$ is a contraction mapping.  To show
that $\mathbf{G}$ maps from $\Omega$ into $\Omega$, suppose $\mathbf{q}\in
\Omega$. Since $\Omega$ contains the fixed-point $\mathbf{q}^*$, it is then
seen that
\begin{align}
  \vectornorm{\mathbf{G}(\mathbf{q})-
    \mathbf{q}^*}_{\infty}&=\vectornorm{\mathbf{G}(\mathbf{q})-
    \mathbf{G}(\mathbf{q}^*)}_{\infty}\\
  &< \vectornorm{\mathbf{q}-\mathbf{q}^*}_{\infty} < r,
\end{align}
and hence $\mathbf{G}(\mathbf{q})\in \Omega$. By invoking the Banach fixed-point
theorem, the fixed-point iteration defined by
\eqref{eq:fixed_point_iter} converges to a unique fixed-point $\mathbf{q}^*$.

We finally observe that, since the fixed-point is unique, the minimizer of the
Lagrangian function $\mathcal{L}$ is unique, and hence the optimal $P_{T|X}$ of
the primal and the dual optimization problem coincide, thus concluding the proof.

\ifCLASSOPTIONcaptionsoff
  \newpage
\fi



\bibliography{tcom}

\bibliographystyle{IEEEtranTCOM}
%
%

%




\end{document}